\documentclass[preprint,1p,11pt]{IR-Template/ISAS_IR}

\pdfoutput=1

\usepackage[latin1]{inputenc}
\usepackage{calc}
\usepackage{stfloats} 
\usepackage{flushend}
\usepackage{calc}
\usepackage{microtype}
\usepackage{amsmath}
\usepackage{amssymb}
\usepackage{graphicx}
\usepackage{psfrag}
\usepackage{verbatim}
\usepackage{multirow}
\usepackage{algorithm2e}
\usepackage{blkarray}
\usepackage{lmodern}
\usepackage{wrapfig}
\usepackage{paralist}
\usepackage{trfsigns}
\usepackage{version}

\date{}

\def\vec#1{\underline{#1}}

\def\1_2{{\frac{1}{2}}}

\newcommand{\rv}[1]{\ensuremath{\boldsymbol{#1}}}

\DeclareMathOperator{\xlog}{xlog}
\DeclareMathOperator{\Ei}{Ei}

\def\SND{standard Normal distribution}

\def\NewR{\mathbb{R}} 

\newcommand{\GDs}{Gaussian densities}
\newcommand{\GM}{Gaussian mixture}
\newcommand{\GMD}{\GM{} density}

\newcommand{\DC}{Dirac component}
\newcommand{\DM}{Dirac mixture}
\newcommand{\DMC}{\DM{} component}
\def\DMA{\DM{} approximation}
\def\DMD{\DM{} density}
\def\DMDs{\DM{} densities}

\def\LCD{Localized Cumulative Distribution}
\def\CvM{Cram\'{e}r-von Mises}
\def\CvMD{\CvM{} Distance}

\def\Eq#1{(\ref{#1})}

\def\Sec#1{Sec.~\ref{#1}}

\def\Fig#1{Fig.~\ref{#1}}

\def\Appendix#1{Appendix~\ref{#1}}

\def\Theorem#1{Theorem~\ref{#1}}

\sfcode`\0=1001
\sfcode`\1=1001
\sfcode`\2=1001
\sfcode`\3=1001
\sfcode`\4=1001
\sfcode`\5=1001
\sfcode`\6=1001
\sfcode`\7=1001
\sfcode`\8=1001
\sfcode`\9=1001
\sfcode`\.=1001
\sfcode`\?=1001
\sfcode`\!=1001
\sfcode`\:=1001

\newcommand\AdaptiveCapitalization[1]{\ifnum\ifhmode\spacefactor\else2000\fi>1000 \uppercase{#1}\else#1\fi}

\newcommand{\cM}{\AdaptiveCapitalization{c}ircular moment}
\newcommand{\vMd}{\AdaptiveCapitalization{v}on Mises distribution}
\newcommand{\WNd}{\AdaptiveCapitalization{w}rapped Normal distribution}
\newcommand{\pdf}{\AdaptiveCapitalization{p}robability density function}
\newcommand{\cpdf}{\AdaptiveCapitalization{c}ircular probability density function}

\usepackage{color}

\def\Box#1{{\fboxsep2pt\fbox{$#1$}}}

\newlength\EqLen

\def\ScaleInner#1{%
\settowidth{\EqLen}{#1}
\ifdim\EqLen < \columnwidth%
\begin{equation*}%
\begin{minipage}{\EqLen}#1\end{minipage}%
\end{equation*}%
\else%
\begin{equation*}%
\resizebox{0.99\columnwidth}{!}{\begin{minipage}{\EqLen}#1\end{minipage}}%
\end{equation*}%
\fi%
}%

\def\Scale#1
  {
  \ScaleInner{$ #1 $} 
  }
  
\def\LongVersion#1{}

\def\citep#1{(\cite{#1})}

\usepackage{amsthm}

\newtheorem{theorem}{Theorem}[section]

\newtheorem{corollary}[theorem]{Corollary}
\theoremstyle{definition}
\newtheorem{definition}{Definition}[section]

\theoremstyle{remark}
\newtheorem{remark}{Remark}[section]

\begin{document}

\begin{frontmatter}

\title{Optimal Reduction of Multivariate\\Dirac Mixture Densities}

\author{Uwe~D.~Hanebeck}
\ead{uwe.hanebeck@ieee.org}

\address{Intelligent Sensor-Actuator-Systems Laboratory (ISAS)\\
Institute for Anthropomatics and Robotics\\
Karlsruhe Institute of Technology (KIT), Germany}

\begin{abstract}
%
%
This paper is concerned with the optimal approximation of a given multivariate Dirac mixture, i.e., a density comprising weighted Dirac distributions on a continuous domain, by an equally weighted Dirac mixture with a reduced number of components.
%
%
The parameters of the approximating density are calculated by minimizing a smooth global distance measure, a generalization of the well-known \CvMD{} to the multivariate case. 
%
%
This generalization is achieved by defining an alternative to the classical cumulative distribution, the \LCD{} (LCD), as a characterization of discrete random quantities (on continuous domains), which is unique and symmetric also in the multivariate case. 
%
%
The resulting approximation method provides the basis for various efficient nonlinear state and parameter estimation methods.
\end{abstract}

\end{frontmatter}

%
%

\section{Introduction} \label{Sec_Introduction}

%
%

\subsection{Motivation} \label{SubSec_Motivation}

	%
%

We consider point sets ${\cal P}_x=\{\vec{x}_1, \vec{x}_2, \ldots, \vec{x}_L\}$ where the points $\vec{x}_i$, $i=1, \ldots, L$ are arbitrarily placed in $\NewR^N$, i.e., $\vec{x}_i \in \NewR^N$, $i=1, \ldots, L$. The points can be of arbitrary origin, e.g., be samples from a \pdf{}, and can be interpreted as being equally weighted or be equipped with weights.
%
%
When they are weighted, we assume the weights $w_i$ associated with point $\vec{x}_i$ for $i=1, \ldots, L$ to be positive and sum up to one.

%
%
The point sets are interpreted as discrete \pdf s over a continuous domain, where the individual points correspond to locations of Dirac distributions with associated weights. The point set will be called a \DMD{}.
%
%
\DMDs{} are a popular representation of densities in stochastic nonlinear filters such as particle filters.
%
%
They characterize random vectors in a similar way as sets of sample points by having a large number of components with large weights in regions of high density and a small number of components with small weights in regions of low density.
%
%
Hence, approximating one \DMD{} by another one while maintaining the information content is equivalent to maintaining its probability mass distribution. 

%
%
Reducing the size of a point set while maintaining its information content as much as possible is a fundamental problem occurring in many flavors in different contexts.
%
%
%
One example is a large number of noisy samples from an underlying unknown probability density function. The goal is to represent the underlying density by a reduced set of well-placed Dirac components while removing the noise.
%
%
The need for reduction also occurs, when the given point set is too complex for further processing. During the recursive processing of Dirac mixtures, for example in a nonlinear filter, the number of components often explodes. A typical example is the propagation of Dirac mixture densities through a discrete-time dynamic system, which requires a kind of generalized convolution with the noise density. After one processing step, the number of components is given by the product of the given number of components times the number of components used for representing the noise density, which results in an exponential increase with the number of processing steps. Hence, the goal is to keep the number of components at a manageable level by performing regular reductions.

%
%

\subsection{Related Work} \label{SubSec_Prior_Work}

	We will now take a look at different methods that have been devised in the general context of reduction of point sets or densities.

%
%

\paragraph{Random Selection}

The most common technique for reducing the number of components of a given  \DMD{} is the random selection of certain components.
%
%
It is commonly used in the prediction step of Particle Filters, where each prior sample is perturbed with a single sample from the noise distribution before propagation thorough the system model. The perturbation can be viewed as generating the noise samples at once with a subsequent random selection from the Cartesian product of prior samples and noise samples.

%
%

\paragraph{Intermediate Continuous Densities}

Another common technique is to replace the given \DM{} by a suitable continuous density in a first step. In a second step, the desired number of samples is drawn from the continuous density. With this technique, it is also possible to increase the number of components as required. However, the first step is equivalent to density estimation from samples, which is by itself a complicated task and an active research topic. Furthermore, this reduction technique introduces undesired side information via the choice of the continuous smoothing density.

%
%

\paragraph{Clustering or Vector Quantization Methods}

Clustering or vector quantization methods also aim at representing a point set by a smaller set of representatives. 
%
%
For optimization purposes, a distortion measure is typically used, which sums up the (generalized) distances between the points and their representatives.
%
%
Minimizing the distortion measure results in two conditions: 1.~Points are associated to their closest representative. 2.~The representative is calculated as the average of all its associated points. 
%
%
As no closed-form solution for performing the minimization of the distortion measure exist, robust iterative procedures have been devised starting with Lloyd's algorithm proposed in 1957 and published later in \cite{lloyd_least_1982}, first called k-means algorithm in \cite{macqueen_methods_1967}, and its extension in the form of the  Linde-Buzo-Gray-algorithm \cite{linde_algorithm_1980}.
%
%
Obviously, the representatives fulfilling the above two conditions do not necessarily maintain the form of the density, which will also be shown by some examples in \Sec{Sec_Numerical} of this paper.
%
%
An additional problem of clustering or vector quantization methods is that the iterative minimization procedures typically get stuck in local minima. Intuitively, the resulting samples are only influenced by samples in the corresponding part of the Voronoi diagram, while the proposed method is based upon a global distance measure.

%
%

\paragraph{Reapproximating Continuous Mixtures with Continuous Mixtures}

%
%

As \DM{} reduction is a special case of general mixture reduction techniques. As these techniques are usually focused on continuous densities such as \GM s, 
e.g., see \cite{crouse_look_2011}, it is worthwhile to discuss the differences. First, when continuous mixtures are re-approximated with continuous mixtures, the densities or parts of the densities can be directly compared in terms of the integral squared difference \cite{williams_cost-function-based_2003} or the Kullback-Leibler divergence \cite{runnalls_kullback-leibler_2007}. Directly comparing densities with an integral measure is not possible when at least one of the densities is a \DMD{} \cite{MFI08_Hanebeck-LCD}. Instead, cumulative distributions can be used in the scalar case or appropriate generalizations for the multivariate case \cite{MFI08_Hanebeck-LCD}. Second, for continuous mixtures two or more critical components can be merged in order to locally reduce the number of components \cite{west_approximating_1993}, where different criteria for identifying components are possible such as small weights. These components are then replaced by a new component with appropriate parameters, e.g., maintaining mean and covariance. Locally replacing components is not straightforward for \DMD s as it is i)~difficult to identify potential merging candidates and ii)~a single replacement component does not capture the extent covered by the original components. Hence, a replacement of several \DC s by a smaller set of \DC s with a cardinality larger than one would be in order. 

%
%

\paragraph{Reapproximating Continuous Mixtures with Discrete Mixtures}

%
%

The reduction problem can be viewed as approximating a given (potentially continuous) density with a \DMD{}. Several options are available for performing this approximation. Moment-based approximations have been proposed in the context of Gaussian densities  and Linear Regression Kalman Filters (LRKFs), see \cite{lefebvre_linear_2005}. Examples are the Unscented Kalman Filter (UKF) in \cite{julier_new_2000} and its scaled version in \cite{julier_scaled_2002}, its higher-order generalization in \cite{tenne_higher_2003}, and a generalization to an arbitrary number of deterministic samples placed along the coordinate axes introduced in \cite{IFAC08_Huber}. 
%
%
For \cpdf s, a first approach to \DMA{} in the vein of the UKF is introduced in \cite{ACC13_Kurz} for the \vMd{} and the \WNd{}. Three components are systematically placed based on matching the first \cM{}. This \DMA{} of continuous \cpdf s has already been applied to sensor scheduling based on bearings-only measurements \cite{Fusion13_Gilitschenski}. In \cite{IPIN13_Kurz}, the results are used to perform recursive circular filtering for tracking an object constrained to an arbitrary one-dimensional manifold.
%
%
For the case that only a finite set of moments of a random vector is given and the underlying density is unknown, an algorithm is proposed in \cite{arXiv14_Hanebeck} for calculating multivariate \DMDs{} with an arbitrary number of arbitrarily placed components maintaining these moments while providing a homogeneous coverage of the state space. This method could also be used for the reduction problem by calculating the moments of the given point set.
%
%
Methods that are based on distance measures between the given density and its \DMA{} have been proposed for the case of scalar continuous densities in \cite{CDC06_Schrempf-DiracMixt, MFI06_Schrempf-CramerMises}. They are based on distance measures between cumulative distribution functions. These distance-based approximation methods are generalized to the multi-dimensional case by defining an alternative to the classical cumulative distribution, the \LCD{} (LCD) \cite{MFI08_Hanebeck-LCD}, which is unique and symmetric. Based on the LCD, multi-dimensional Gaussian densities are approximated by \DMDs{} in \cite{CDC09_HanebeckHuber}. A more efficient method for the case of \SND s with a subsequent transformation is given in \cite{ACC13_Gilitschenski}.

%
%

The LCD-based methods will be extended to the reduction of \DMDs{} in this paper. A variant of the reduction problem, the optimal approximation of the Cartesian product of marginal \DMDs{} is considered in \cite{ACC10_Eberhardt} and a solution is proposed that does not require the explicit calculation of all combinations of marginal components.

%
%

\subsection{Key Ideas and Results of the Paper} \label{SubSec_KeyIdea}

	%
%
The key idea of this paper is the systematic reapproximation of \DMDs{} by minimization of a novel distance measure. The distance measure compares the probability masses of both densities under certain kernels for all possible kernel locations and widths, which allows the use of integral measures for the mass functions.
%
%
This approximation method is similar to the approximation of multivariate \GDs{} by \DM s in \cite{CDC09_HanebeckHuber}.
%
%
However, calculating the distance measure between multivariate Gaussians and \DMDs{} in \cite{CDC09_HanebeckHuber} requires a one-dimensional numerical approximation, while the distance measure for comparing \DMDs{} with \DMDs{} proposed in this paper is given in closed from.

%
%

%
%
The resulting distance measure is smooth and does not suffer from local minima, so that standard optimization methods can be used for calculating the desired \DMA{}. 
%
%
As no randomness is involved, the optimization results are completely deterministic and reproducible, which is in contrast to random selection procedures and most clustering methods.

%
%

The results for approximating $2000$ samples from a \SND{} by a \DMA{} with $L=10$, $L=20$, and $L=30$ components are shown in \Fig{Fig_Gaussian_DMA}.

\begin{figure*}[t]
	\includegraphics{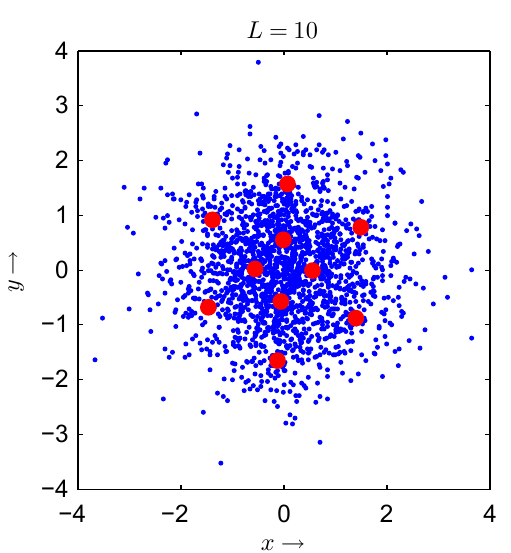}
	\hspace*{\fill}
	\includegraphics{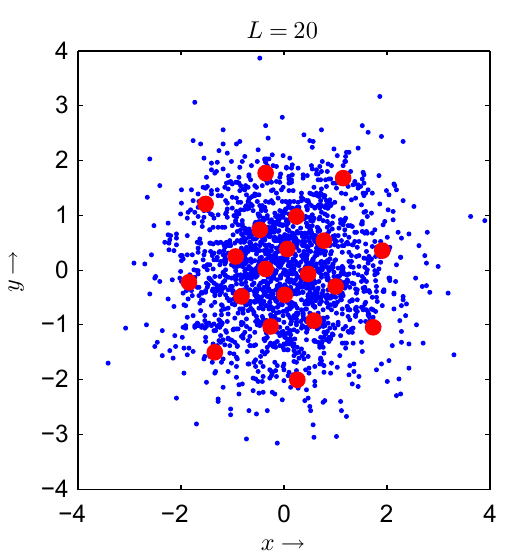}
	\hspace*{\fill}
	\includegraphics{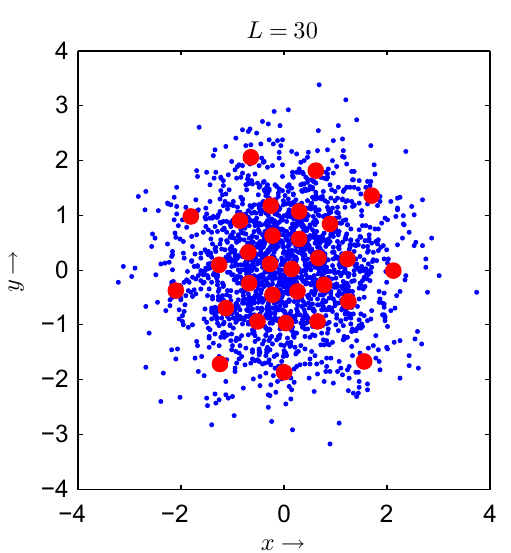}
	\caption{\DMA{} of $2000$ samples from a \SND{} with $L=10$, $L=20$, and $L=30$.  Blue: Random samples representing the \SND{}. Red: Reduced point set.}
	\label{Fig_Gaussian_DMA}
\end{figure*}

%
%

\subsection{Organization of the Paper} \label{SubSec_Organization}

	%
%
In the next section, a rigorous formulation of the considered approximation problem is given. 
%
%
For comparing \DMDs{}, an alternative to the classical cumulative distribution, the so called Localized Cumulative Distribution (LCD) is introduced in \Sec{Sec_LCD}. 
%
%
Based on this LCD, a generalization of the \CvMD{}, which is the squared integral difference between the LCD of the given density and the LCD of the approximate density is given in \Sec{Sec_CvD}. This new distance measure is used for analysis purposes, i.e., for comparing the approximate Dirac mixture to the given one. 
%
%
The synthesis problem, i.e., determining the parameters of the approximate Dirac mixture in such a way that it is as close as possible to the given Dirac mixture according to the new distance measure is the topic of \Sec{Sec_Reduction}. Minimization is performed with a quasi-Newton method. The required gradient is derived in \Appendix{Sec_Gradient}, which also gives necessary conditions for a minimum distance in the form of a set of nonlinear equations.
%
%
The complexity of both calculating the distance measure and its gradient in practical situations is derived in \Sec{Sec_Complexity}.
%
%
Examples of using the new reduction method on specific point sets are given in \Sec{Sec_Numerical}.
%
%
The new approach is discussed in \Sec{Sec_Conclusions} and an outlook to future work is given.

%
%

\section{Problem Formulation} \label{Sec_Prob_Form}

	We consider an $N$--dimensional Dirac mixture density with $M$ components given by
\begin{equation*}
\tilde{f}(\vec{x}) = \sum_{i=1}^M w_i^y \, \delta(\vec{x} - \vec{y}_i) \enspace ,
\end{equation*}
with positive weights $w_i^y>0$, i.e., $w_i^y > 0$, for $i=1, \ldots, M$, that sum up to one and $M$ locations $\vec{y}_i = \begin{bmatrix} y_i^{(1)}, y_i^{(2)}, \ldots, y_i^{(N)} \end{bmatrix}^T$ for $i=1, \ldots, M$. This density is approximated by another $N$--dimensional Dirac mixture density with $L$ components given by
\begin{equation}\label{Eq_DiracMixture}
f(\vec{x}) = \sum_{i=1}^L w_i^x \, \delta(\vec{x} - \vec{x}_i) \enspace ,
\end{equation}
with positive weights $w_i^x$, i.e., $w_i^x > 0$, for $i=1, \ldots, L$, that sum up to one and $L$ locations $\vec{x}_i = \begin{bmatrix} x_i^{(1)}, x_i^{(2)}, \ldots, x_i^{(N)} \end{bmatrix}^T$ for $i=1, \ldots, L$, where we typically assume $L \le M$.

%
%
The goal is to select the location parameters $\vec{x}_i$, $i=1, \ldots, L$ of the approximating density $f(\vec{x})$ in such a way that a distance measure $D$ between the true density $\tilde{f}(\vec{x})$ and its approximation $f(\vec{x})$ is systematically minimized. The weights  $w_i^x$, $i=1, \ldots, L$ are assumed to be given and are typically set to be equal. An extension to given unequal weights or to even optimizing the weights of $f(\vec{x})$ is a simple extension that is not pursued in this paper.

%
%

The true \DMD{} might already have equally weighted components, so that the information is solely stored in the component locations. In this case, the goal of the approximation is a pure reduction of the number of components. On the other hand, the components of the true \DMD{} might have different weights. This could be the result of, e.g., weighting a prior \DMD{} by a likelihood function in a Bayesian filtering setup. In that case, the approximation replaces an arbitrarily weighted \DMD{} by an equally weighted one.
%
%
In the latter case, an equal number of components, i.e., $L=M$, can be useful.


\section{Localized Cumulative Distribution} \label{Sec_LCD}

    %
%

For the systematic reduction of the number of components of a given \DMD{}, a distance measure for comparing the original density and its approximation is required. However, \DMDs{} cannot be directly compared as they typically do not even share a common support. 
%
%
Typically, their corresponding cumulative distributions are used for comparison purposes,
%
%
as is the case in certain statistical tests such as the Kolmogov-Smirnov test \cite[p.~623]{press_numerical_1992}.
%
%
However, it has been shown in \cite{MFI08_Hanebeck-LCD} that although the cumulative distribution is well suited for comparing scalar densities, it exhibits several problems in higher-dimensional spaces: It is non-unique and non-symmetric.
%
%
In addition, integral measures for comparing two cumulative distributions do not converge over infinite integration domains when the underlying \DMDs{} differ.

%
%
As an alternative transformation of densities, the Localized Cumulative Distribution (LCD) introduced in \cite{MFI08_Hanebeck-LCD} is employed here in a generalized form.
%
%
An LCD is an integral measure proportional to the mass concentrated in a region with a size parametrized by a vector $\vec{b}$ around test points $\vec{m}$. These regions are defined by kernels $K(\vec{x}-\vec{m}, \vec{b})$ centered around $\vec{m}$ with size $\vec{b}$.
\begin{definition}
Let $\vec{\rv{x}}$ be a random vector with $\vec{\rv{x}} \in \NewR^N$, which is characterized by an $N$--dimensional probability density function $f: \NewR^N\!\!\rightarrow\!\!\NewR_{+}$. The corresponding Localized Cumulative Distribution (LCD) is defined as
\begin{equation*}
F(\vec{m}, \vec{b}) = \int _{\NewR^N} f(\vec{x}) \, K(\vec{x}-\vec{m}, \vec{b}) \, d \vec{x}
\end{equation*}
with $\vec{b} \in \NewR_{+}^N$ and $F(.,.): \Omega \rightarrow [0,1], \Omega \subset \NewR^N \times \NewR_{+}^N$.
\end{definition}

\begin{definition}

As a shorthand notation, we will denote the relation between the density $f(\vec{x})$ and its LCD $F(\vec{x}, \vec{b})$ by
\begin{equation*}
f(\vec{x}) \; \laplace \; F(\vec{m}, \vec{b}) \enspace .
\end{equation*}
\end{definition}

In this paper, we focus attention on separable kernels of the type
\begin{equation*}
K(\vec{x}-\vec{m}, \vec{b}) = \prod_{k=1}^N K(x^{(k)} - m^{(k)}, b^{(k)}) \enspace .
\end{equation*}
Furthermore, we consider kernels with equal width in every dimension, i.e., $b_i=b$ for $i=1, \ldots, N$, which gives
\begin{equation*}
K(\vec{x}-\vec{m}, b) = \prod_{i=k}^N K(x^{(k)} - m^{(k)}, b) \enspace .
\end{equation*}

%
%
Rectangular, axis-aligned kernels as used in \cite{MFI08_Hanebeck-LCD} are the obvious choice as they yield the probability mass of the considered density in a rectangular region centered around $\vec{m}$.
%
%
Rectangular kernels are a good choice for analysis purposes and are used, e.g., when assessing the discrepancy of a point set from a uniform distribution.

%
%
However, for synthesizing a suitable approximation for a given (nonuniform) \DM{} with a smaller number of components, smooth kernels lead to simpler optimization problems. Here, we consider kernels of Gaussian type according to
\begin{equation*}
K(\vec{x}-\vec{m}, b) = \prod_{k=1}^N \exp\left( -\frac{1}{2} \frac{\left(x^{(k)}-m^{(k)}\right)^2}{b^2} \right) \enspace .
\end{equation*}
Based on a Gaussian kernel, an N-dimensional Dirac component $\delta(\vec{x} - \hat{\vec{x}})$ at location $\hat{\vec{x}}$ corresponds to its LCD $\Delta(\vec{m}, b)$
\begin{equation*}
\delta(\vec{x} - \hat{\vec{x}}) \; \laplace \; \Delta(\vec{m}, b)
\end{equation*}
with
\begin{equation*}
\begin{split}
\Delta(\vec{m}, b) & = \int _{\NewR^N} \delta(\vec{x} - \hat{\vec{x}}) \, K(\vec{x}-\vec{m}, \vec{b}) \, d \vec{x}\\
& = \prod_{k=1}^N \exp\left( -\frac{1}{2} \frac{\left(\hat{x}^{(k)}-m^{(k)}\right)^2}{b^2} \right) \enspace .
\end{split}
\end{equation*}
%
%
With this LCD of a single Dirac component, the LCD of the \DM{} in \Eq{Eq_DiracMixture} is given by
\begin{equation*}
F(\vec{m}, b) = \sum_{i=1}^L w_i^x \, \prod_{k=1}^N \exp\left( -\frac{1}{2} \frac{\left(x_i^{(k)}-m^{(k)}\right)^2}{b^2} \right) \enspace .
\end{equation*}
A similar result holds for the original \DM{} $\tilde{f}(\vec{x})$.

    
\section{A Modified \CvMD{}} \label{Sec_CvD}

	The \LCD{} (LCD) defined previously can now be used to derive a modified version of the \CvMD{} suitable for comparing Dirac Mixtures. This new distance is defined as the integral of the square of the difference between the LCD of the true density $\tilde{f}(\vec{x})$ and the LCD of its approximation $f(\vec{x})$.

\begin{definition}[Modified \CvMD{}]
\label{Def_ModCvMDist}%
The distance $D$ between two densities $\tilde{f}(\vec{x}): \NewR^N \rightarrow \NewR_{+}$ and $f(\vec{x}): \NewR^N \rightarrow \NewR_{+}$ is given in terms of their corresponding LCDs $\tilde{F}(\vec{x}, b)$ and $F(\vec{x}, b)$ as
\begin{equation*}
D = \int _{\NewR_{+}} w(b) \int _{\NewR^N} \left( \tilde{F}(\vec{m}, b) - F(\vec{m}, b) \right )^2 \, d \vec{m} \, d b \enspace ,
\end{equation*}
where $w(b): \NewR_{+} \rightarrow [0,1]$ is a suitable weighting function.
\end{definition}
A weighting function $w(b)$ has been introduced that controls how kernels of different sizes influence the resulting distance, which provides some degrees of freedom during the design of an approximation algorithm. Alternatively, a unit weighting function could be used while modifying the kernels accordingly.
\begin{theorem} \label{Theorem_General_Distance_Measure}
By inserting the LCDs
\begin{equation*}
\tilde{F}(\vec{m}, b) = \sum_{i=1}^M w_i^{y} \prod_{k=1}^N \exp\left( -\frac{1}{2} \frac{(y_i^{(k)}-m^{(k)})^2}{b^2} \right)
\end{equation*}
and
\begin{equation*}
F(\vec{m}, b) = \sum_{i=1}^L w_i^{x} \prod_{k=1}^N \exp\left( -\frac{1}{2} \frac{(x_i^{(k)}-m^{(k)})^2}{b^2} \right) \enspace ,
\end{equation*}
and by using the weighting function
\begin{equation*}
w(b) = \begin{cases} \frac{1}{b^{N-1}} & b \in [0, b_{\text{max}}] \\ 0 & \text{elsewhere} \end{cases}
\end{equation*}
the following expressions for the distance $D$
\begin{equation*}
\begin{split}
D = 
\sum_{i=1}^M \sum_{j=1}^M w_i^{y} w_j^{y} & \, \gamma \left( \displaystyle\sum_{k=1}^N \left( y_i^{(k)} - y_j^{(k)} \right)^2 \right) \\
- 2 \sum_{i=1}^L \sum_{j=1}^M w_i^{x} w_j^{y} & \, \gamma \left( \displaystyle\sum_{k=1}^N \left( x_i^{(k)} - y_j^{(k)} \right)^2 \right) \\
+ \sum_{i=1}^L \sum_{j=1}^L w_i^{x} w_j^{x} & \, \gamma \left( \displaystyle\sum_{k=1}^N \left( x_i^{(k)} - x_j^{(k)} \right)^2 \right)
\end{split}
\end{equation*}
with
\begin{equation*}
\gamma(z) = \frac{\pi^{\frac{N}{2}}}{8} \left\{ 4 \, b_\text{max}^2 \, \exp\left( -\frac{1}{2} \frac{ z }{2 \, b_\text{max}^2} \right) + z \,%
\Ei\left( -\frac{1}{2} \frac{ z }{2 \, b_\text{max}^2} \right) \right\} \enspace ,%
\end{equation*}
are obtained, where $\Ei(z)$ denotes the exponential integral.
\end{theorem}
\begin{proof}
For the given specific weighting function $w(b)$, the distance measure is given by
\begin{equation}
D = \int _0^{b_\text{max}} \frac{1}{b^{N-1}} \int _{\NewR^N} \left( \tilde{F}(\vec{m}, b) - F(\vec{m}, b) \right )^2 \, d \vec{m} \, d b \enspace .
\label{Eq_Distance_Measure_bmax}
\end{equation}
Inserting the LCDs $\tilde{F}(\vec{m}, b)$ and $F(\vec{m}, b)$ leads to
\begin{equation*}
\begin{split}
D = \int _0^{b_\text{max}} \frac{1}{b^{N-1}} \int _{\NewR^N}
\sum_{i=1}^M \sum_{j=1}^M w_i^{y} w_j^{y} 
& \prod_{k=1}^N \exp\left( -\frac{1}{2} \frac{(y_i^{(k)}-m^{(k)})^2}{b^2} \right) \\
& \prod_{k=1}^N \exp\left( -\frac{1}{2} \frac{(y_j^{(k)}-m^{(k)})^2}{b^2} \right) \\
- 2 \sum_{i=1}^L \sum_{j=1}^M w_i^{x} w_j^{y} 
& \prod_{k=1}^N \exp\left( -\frac{1}{2} \frac{(x_i^{(k)}-m^{(k)})^2}{b^2} \right) \\
& \prod_{k=1}^N \exp\left( -\frac{1}{2} \frac{(y_j^{(k)}-m^{(k)})^2}{b^2} \right) \\
+ \sum_{i=1}^L \sum_{j=1}^L w_i^{x} w_j^{x} 
& \prod_{k=1}^N \exp\left( -\frac{1}{2} \frac{(x_i^{(k)}-m^{(k)})^2}{b^2} \right) \\
& \prod_{k=1}^N \exp\left( -\frac{1}{2} \frac{(x_j^{(k)}-m^{(k)})^2}{b^2} \right)
\, d \vec{m} \, d b \enspace .
\end{split}
\end{equation*}
Exchanging integration and summation gives
\begin{equation*}
\begin{split}
D = 
\sum_{i=1}^M \sum_{j=1}^M w_i^{y} w_j^{y} & \int _0^{b_\text{max}} \frac{1}{b^{N-1}} \prod_{k=1}^N \int _{\NewR}
\exp\left( -\frac{1}{2} \frac{(y_i^{(k)}-m^{(k)})^2}{b^2} \right)
\exp\left( -\frac{1}{2} \frac{(y_j^{(k)}-m^{(k)})^2}{b^2} \right) d m^{(k)} \, d b\\
- 2 \sum_{i=1}^L \sum_{j=1}^M w_i^{x} w_j^{y} & \int _0^{b_\text{max}} \frac{1}{b^{N-1}} \prod_{k=1}^N \int _{\NewR}
\exp\left( -\frac{1}{2} \frac{(x_i^{(k)}-m^{(k)})^2}{b^2} \right)
\exp\left( -\frac{1}{2} \frac{(y_j^{(k)}-m^{(k)})^2}{b^2} \right) d m^{(k)} \, d b\\
+ \sum_{i=1}^L \sum_{j=1}^L w_i^{x} w_j^{x} & \int _0^{b_\text{max}} \frac{1}{b^{N-1}} \prod_{k=1}^N \int _{\NewR}
\exp\left( -\frac{1}{2} \frac{(x_i^{(k)}-m^{(k)})^2}{b^2} \right)
\exp\left( -\frac{1}{2} \frac{(x_j^{(k)}-m^{(k)})^2}{b^2} \right) d m^{(k)} \, d b
\end{split}
\end{equation*}
For further simplification, the following closed-form expression for the occurring integrals
\begin{equation}
\int _{\NewR} \exp\left( -\frac{1}{2} \frac{(z_i - m)^2}{b^2} \right) \exp\left( -\frac{1}{2} \frac{(z_j - m)^2}{b^2} \right) d m
= \sqrt{\pi} \, b \, \exp \left( -\frac{1}{2} \frac{(z_i - z_j)^2}{2 \, b^2} \right) \enspace ,
\label{Eq_Simplification_Product_of_Exp}
\end{equation}
is used. This gives
\begin{equation*}
\begin{split}
D = 
\sum_{i=1}^M \sum_{j=1}^M w_i^{y} w_j^{y} & \int _0^{b_\text{max}} b \, \pi^{\frac{N}{2}} \prod_{k=1}^N
\exp\left( -\frac{1}{2} \frac{ \left( y_i^{(k)} - y_j^{(k)} \right)^2}{2 \, b^2} \right) \, d b\\
- 2 \sum_{i=1}^L \sum_{j=1}^M w_i^{x} w_j^{y} & \int _0^{b_\text{max}} b \, \pi^{\frac{N}{2}} \prod_{k=1}^N
\exp\left( -\frac{1}{2} \frac{ \left( x_i^{(k)} - y_j^{(k)} \right)^2}{2 \, b^2} \right) \, d b\\
+ \sum_{i=1}^L \sum_{j=1}^L w_i^{x} w_j^{x} & \int _0^{b_\text{max}} b \, \pi^{\frac{N}{2}} \prod_{k=1}^N
\exp\left( -\frac{1}{2} \frac{ \left( x_i^{(k)} - x_j^{(k)} \right)^2}{2 \, b^2} \right) \, d b
\enspace .
\end{split}
\end{equation*}
or
\begin{equation*}
\begin{split}
D = 
\sum_{i=1}^M \sum_{j=1}^M w_i^{y} w_j^{y} & \int _0^{b_\text{max}} b \, \pi^{\frac{N}{2}}
\exp\left( -\frac{1}{2} \frac{ \displaystyle\sum_{k=1}^N \left( y_i^{(k)} - y_j^{(k)} \right)^2}{2 \, b^2} \right) \, d b\\
- 2 \sum_{i=1}^L \sum_{j=1}^M w_i^{x} w_j^{y} & \int _0^{b_\text{max}} b \, \pi^{\frac{N}{2}}
\exp\left( -\frac{1}{2} \frac{ \displaystyle\sum_{k=1}^N \left( x_i^{(k)} - y_j^{(k)} \right)^2}{2 \, b^2} \right) \, d b\\
+ \sum_{i=1}^L \sum_{j=1}^L w_i^{x} w_j^{x} & \int _0^{b_\text{max}} b \, \pi^{\frac{N}{2}}
\exp\left( -\frac{1}{2} \frac{ \displaystyle\sum_{k=1}^N \left( x_i^{(k)} - x_j^{(k)} \right)^2}{2 \, b^2} \right) \, d b
\enspace .
\end{split}
\end{equation*}
With
\begin{equation*}
\int _0^{b_\text{max}} b \, \exp\left( -\frac{1}{2} \frac{ z }{2 \, b^2} \right) \, d b
= \frac{1}{8} \left\{ 4 \, b_\text{max}^2 \, \exp\left( -\frac{1}{2} \frac{ z }{2 \, b_\text{max}^2} \right) + z \,
\Ei\left( -\frac{1}{2} \frac{ z }{2 \, b_\text{max}^2} \right\}
 \right)
\end{equation*}
for $z>0$, the final result is obtained.
\end{proof}
\begin{remark}
The exponential integral $\Ei(z)$ is defined as
\begin{equation*}
\Ei(z) = \int_{-\infty}^z \frac{e^t}{t} d t
\enspace .
\end{equation*}
For $z>0$, $\Ei(z)$ is related to the incomplete gamma function $\Gamma(0, z)$ according to
\begin{equation*}
\Ei(-z) = - \Gamma(0, z)
\enspace .
\end{equation*}
\end{remark}
\begin{theorem} \label{Theorem_General_Distance_Measure_2}
For large $b_\text{max}$, the distance $D$ is described by
\begin{equation}
D = \frac{\pi^{\frac{N}{2}}}{8} \left( D_{\vec{y}} - 2 D_{\vec{x}\vec{y}} + D_{\vec{x}} \right) + \frac{\pi^{\frac{N}{2}}}{4} \, C_b \, D_{E}
\enspace ,
\label{Equation_Approximate_Distance}
\end{equation}
with the constant $C_b = \log(4 \, b_\text{max}^2) - \Gamma$. Here only the last term depends upon $b_{\text{max}}$ and
\begin{equation*}
D_{\vec{y}} = \sum_{i=1}^M \sum_{j=1}^M w_i^{y} w_j^{y} \, \xlog \left( \sum_{k=1}^N \left( y_i^{(k)} - y_j^{(k)} \right)^2  \right) 
\enspace ,
\end{equation*}
\begin{equation*}
D_{\vec{x}\vec{y}} = \sum_{i=1}^L \sum_{j=1}^M w_i^{x} w_j^{y} \, \xlog \left( \sum_{k=1}^N \left( x_i^{(k)} - y_j^{(k)} \right)^2 \right) 
\enspace ,
\end{equation*}
\begin{equation*}
D_{\vec{x}} = \sum_{i=1}^L \sum_{j=1}^L w_i^{x} w_j^{x} \, \xlog \left( \sum_{k=1}^N \left( x_i^{(k)} - x_j^{(k)} \right)^2 \right) 
\enspace ,
\end{equation*}
with $\xlog(z) = z \cdot \log(z)$ and
\begin{equation*}
D_{E} = \sum_{k=1}^N \left( \sum_{i=1}^L w_i^{x} x_i^{(k)} - \sum_{i=1}^M w_i^{y} y_i^{(k)} \right)^2
\enspace .
\end{equation*}
\end{theorem}

\begin{proof}
For small $z>0$, the exponential integral can be approximated by
\begin{equation}
\Ei(-z) \approx \Gamma + \log(z) - z \enspace ,
\label{Eq_Ei_Approx}
\end{equation}
where $\Gamma \approx 0.5772$ is the Euler gamma constant. As a result, the function $\gamma(z)$ can be approximated according to
\begin{equation*}
\begin{split}
\gamma(z) & \approx \frac{\pi^{\frac{N}{2}}}{8} \left\{ 4 \, b_\text{max}^2 \, \exp\left( -\frac{1}{2} \frac{ z^2 }{2 \, b_\text{max}^2} \right) \right. \\ 
& \left. + z \, \left( \Gamma + \log\left( \frac{1}{2} \frac{ z }{2 \, b_\text{max}^2} \right) - \frac{1}{2} \frac{ z }{2 \, b_\text{max}^2} \right) \right\} \\
& \approx \frac{\pi^{\frac{N}{2}}}{8} \left\{ 4 \, b_\text{max}^2 + z \, \left( \Gamma - \log(4 \, b_\text{max}^2) + \log(z) \right) \right\} \\
& = \frac{\pi^{\frac{N}{2}}}{8} \left\{ 4 \, b_\text{max}^2 - C_b \, z + \xlog(z) \right\}
\enspace .
\end{split}
\end{equation*}
Inserting the first term into the distance measure $D$ in Theorem \ref{Theorem_General_Distance_Measure} cancels due to the fact that
\begin{equation*}
\begin{gathered}
\frac{\pi^{\frac{N}{2}}}{2} b_\text{max}^2 \left\{ \sum_{i=1}^M \sum_{j=1}^M w_i^{y} w_j^{y}
- 2 \sum_{i=1}^L \sum_{j=1}^M w_i^{x} w_j^{y}
+ \sum_{i=1}^L \sum_{j=1}^L w_i^{x} w_j^{x} \right\} \\
= \frac{\pi^{\frac{N}{2}}}{2} b_\text{max}^2 \left\{ \sum_{i=1}^M w_i^{y} - \sum_{i=1}^L w_i^{x} \right\}^2 = 0
\enspace .
\end{gathered}
\end{equation*}
Inserting the second term according to
\begin{equation*}
\begin{split}
- \frac{\pi^{\frac{N}{2}}}{8} C_b \sum_{k=1}^N \Bigg\{
& \sum_{i=1}^M \sum_{j=1}^M w_i^{y} w_j^{y} \left( y_i^{(k)} - y_j^{(k)} \right)^2 \\
- 2 & \sum_{i=1}^L \sum_{j=1}^M w_i^{x} w_j^{y} \left( x_i^{(k)} - y_j^{(k)} \right)^2\\
+ & \sum_{i=1}^L \sum_{j=1}^L w_i^{x} w_j^{x} \left( x_i^{(k)} - x_j^{(k)} \right)^2 \Bigg\}
\enspace ,
\end{split}
\end{equation*}
can be written as
\begin{equation*}
\begin{split}
- \frac{\pi^{\frac{N}{2}}}{8} C_b \sum_{k=1}^N
\Bigg\{ & \sum_{i=1}^M w_i^{y} \left( y_i^{(k)} \right)^2
- 2 \sum_{i=1}^M \sum_{j=1}^M w_i^{y} w_j^{y} y_i^{(k)} y_j^{(k)}
+ \sum_{i=1}^M w_i^{y} \left( y_i^{(k)} \right)^2 \\
-2 \Bigg[ & \sum_{i=1}^L w_i^{x} \left( x_i^{(k)} \right)^2
- 2 \sum_{i=1}^L \sum_{j=1}^M w_i^{x} w_j^{y} x_i^{(k)} y_j^{(k)}
+ \sum_{i=1}^M w_i^{y} \left( y_i^{(k)} \right)^2 \Bigg]\\
+ & \sum_{i=1}^L w_i^{x} \left( x_i^{(k)} \right)^2
- 2 \sum_{i=1}^L \sum_{j=1}^L w_i^{x} w_j^{x} x_i^{(k)} x_j^{(k)}
+ \sum_{i=1}^L w_i^{x} \left( x_i^{(k)} \right)^2 \Bigg\}
\enspace .
\end{split}
\end{equation*}
Canceling corresponding terms finally gives
\begin{equation*}
\frac{\pi^{\frac{N}{2}}}{4} C_b \sum_{k=1}^N
\left( \sum_{i=1}^M w_i^{y} y_i^{(k)} - \sum_{i=1}^L w_i^{x} x_i^{(k)} \right)^2
\enspace .
\end{equation*}
Inserting the third term gives the remaining expressions.
\end{proof}

\begin{remark}
For equal expected values of the densities $\tilde{f}(\vec{x})$ and $f(\vec{x})$, the distance measure in Theorem \ref{Theorem_General_Distance_Measure_2} does not depend upon $b_{\text{max}}$ anymore.
\end{remark}


\section{Reduction} \label{Sec_Reduction}

	%
%

The goal is to find the optimal $L$ locations $\vec{x}_i$, $i=1, \ldots, L$ of the approximating \DMD{} such that the distance measure $D$ in \Eq{Equation_Approximate_Distance} in \Theorem{Theorem_General_Distance_Measure_2} is minimized.
%
%
For optimization, we use a quasi-Newton method, the Broyden-Fletcher-Goldfarb-Shanno (BFGS) algorithm.
%
%
The required gradient $G$ is given in closed form in \Appendix{Sec_Gradient}.

%
%

The final expressions for the distance measure $G$ in \Eq{Equation_Approximate_Distance} in \Theorem{Theorem_General_Distance_Measure_2} and its gradient $G$ in \Eq{Eq_Approximate_Gradient} in \Theorem{Theorem_Approximate_Gradient} do not depend on the maximum kernel width $b_\text{max}$, when the means of the original \DM{} and its reduction are equal. To enforce equal means during the optimization with an unconstrained optimization method, $b_\text{max}$ is set to a large value in these expressions. Alternatively, $b_\text{max}$ could be set to zero in the expressions \Eq{Equation_Approximate_Distance} and \Eq{Eq_Approximate_Gradient} so that they are independent of $b_\text{max}$, while the constraint of equal means is handled by a constrained optimization method.

%
%

Unless prior knowledge about the locations $\vec{x}_i$, $i=1, \ldots, L$ of the approximating \DMD{} is available, the locations are initialized with random samples before starting the optimization.


\section{Complexity} \label{Sec_Complexity}

	%
%

Finding the the minimum of the distance $D$ in \Eq{Equation_Approximate_Distance} either with or without employing the gradient $G$ or the direct solution of the system of nonlinear equations in \Eq{Corollary_System_of_Equations} requires numerical optimization routines with the time complexity depending on the specific routine employed for that purpose. For that reason, the focus will be on analyzing the complexity of performing one evaluation of the distance $D$ in \Eq{Equation_Approximate_Distance} and the corresponding gradient $G$ in \Eq{Eq_Approximate_Gradient} or the equations in \Eq{Corollary_System_of_Equations}.

%
%

The evaluation of the distance $D$ in \Eq{Equation_Approximate_Distance} requires $\mathcal{O}\left( (M^2 + M \cdot L + L^2) \cdot N\right)$ operations, with $M$ the number of \DC s in the original \DM{}, $L$ the number of \DC s used for the approximation. $N$ is the number of dimensions.
%
%
As the first term does not depend upon the desired component locations, it can often be neglected, for example during optimization where only changes of the distance are needed. It is only required when the absolute value of the distance is of interest, e.g., when comparing different approximations.
%
%
As a result, calculating changes of the distance with respect to changes in locations costs $\mathcal{O}\left( (M \cdot L + L^2) \cdot N\right)$ operations. 
%
%
When the number of components $L$ of the approximation is much smaller than the number of \DMC s $M$ of the given original density, i.e., we have $L \ll M$, the complexity of calculating the third term in \Eq{Equation_Approximate_Distance} can be neglected. In that case, we obtain a complexity of $\mathcal{O}\left( M \cdot L \cdot N \right)$ operations, which is linear in $M$, $L$, and $N$.

%
%

Evaluating the necessary conditions for the desired minimum in \Eq{Corollary_System_of_Equations} requires $\mathcal{O}\left( (M \cdot L + L^2) \cdot N\right)$ operations. Again assuming $L \ll M$, this results in a complexity of $\mathcal{O}\left( M \cdot L \cdot N \right)$ operations as for the distance.


\section{Numerical Evaluation} \label{Sec_Numerical}

	%
%

The proposed method for the optimal reduction of \DMDs{} will now be evaluated and compared to a standard clustering technique, the k-means algorithm \cite{macqueen_methods_1967}.

%
%

The results of approximating random samples from a \SND{} have already been shown in \Fig{Fig_Gaussian_DMA} in the introduction. 
%
%
We now consider the reduction of deterministic samples from a \SND{} corrupted by a single outlier.

%
%
In the next step, we approximate samples from a \GMD{}. For that purpose, we generated samples from a \GMD{} with four components, see \Fig{Fig_GM_DMA}. It is important to note that we have a total of $M=4000$ samples, but the number of samples differs for each component: We have $500$ samples for components $(1,1)$ and $(2,2)$ and $1500$ samples for components $(1,2)$ and $(2,1)$. After the reduction from $M=4000$ samples down to $L=40$ samples, we would expect that the probability masses for each component of the \GMD{} are maintained. This is exactly the case for the proposed LCD reduction as can be seen in \Fig{Fig_GM_DMA} on the left side, where we end up with $5$ samples for components $(1,1)$ and $(2,2)$ and $15$ samples for components $(1,2)$ and $(2,1)$. For k-means clustering, shown on the right side in \Fig{Fig_GM_DMA}, this is not the case, so the original distribution is not maintained. In addition, the results of k-means clustering are not reproducible and change with every run.

\begin{figure*}[t]
	\hspace*{\fill}
	\includegraphics{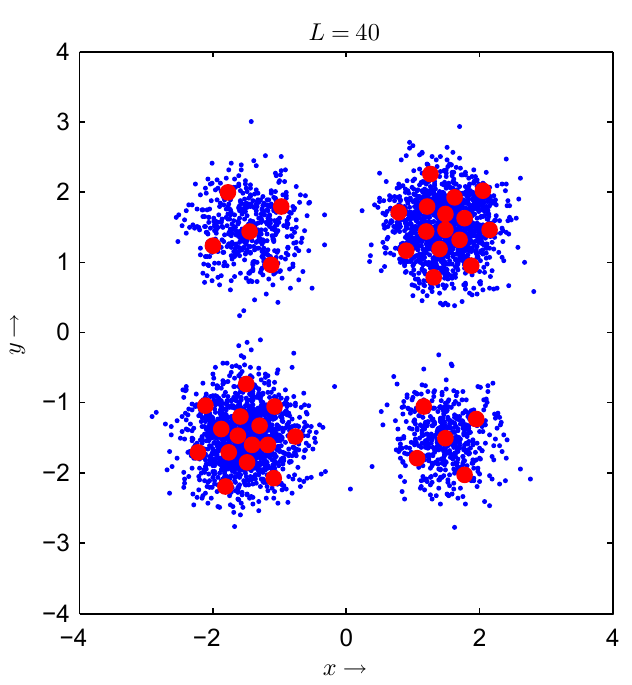}
	\hspace*{\fill}
	\includegraphics{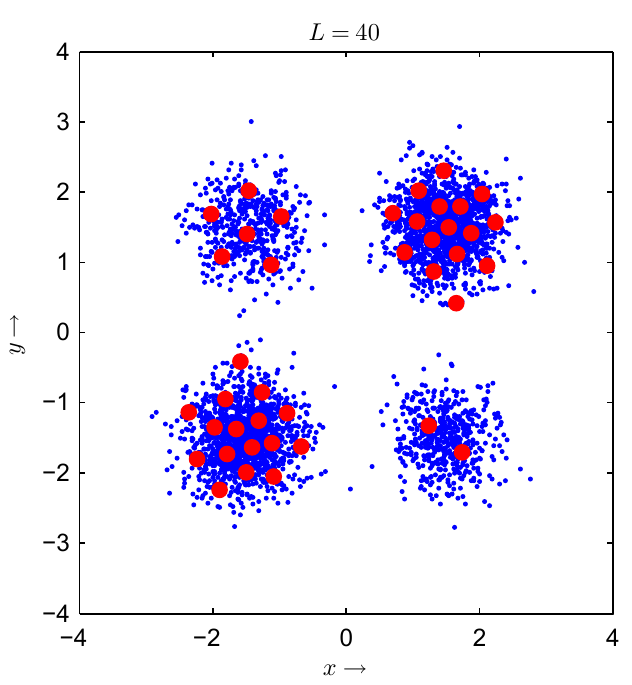}
	\hspace*{\fill}
	\caption{Reduction of a \GMD{} with four components and a varying number of samples per component from $4000$ points to $40$ points. Blue: Random samples representing the \GMD{}. Red: Reduced point set. (left) LCD reduction. (right) Result of k-means clustering.}
	\label{Fig_GM_DMA}
\end{figure*}

%
%

Another way to demonstrate that the proposed reduction method maintains the probability mass distribution is to compare histograms of the samples before and after reduction. To simplify visualization, histograms are calculated for the marginals in $x$-direction. \Fig{Fig_Histograms} shows the histogram of the originals samples on the left side. The histogram after reduction with the proposed LCD reduction method is shown in the middle, while the histogram of the results obtained with k-means are shown on the right side. It is obvious that the histogram of the LCD reduction is much closer to the original histogram than the histogram of k-means.

\begin{figure*}[t]
	\includegraphics{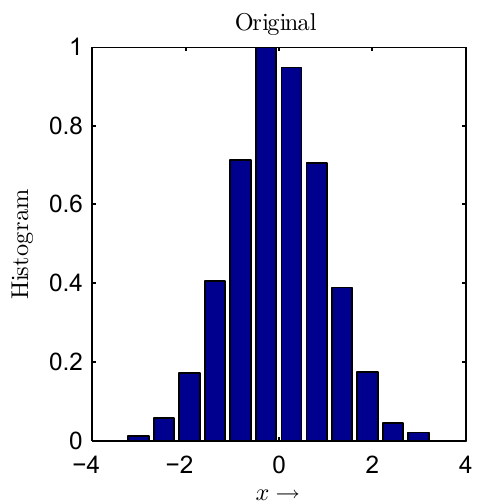}
	\hspace*{\fill}
	\includegraphics{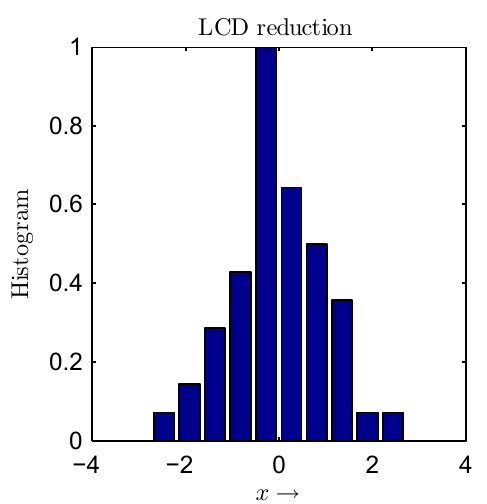}
	\hspace*{\fill}
	\includegraphics{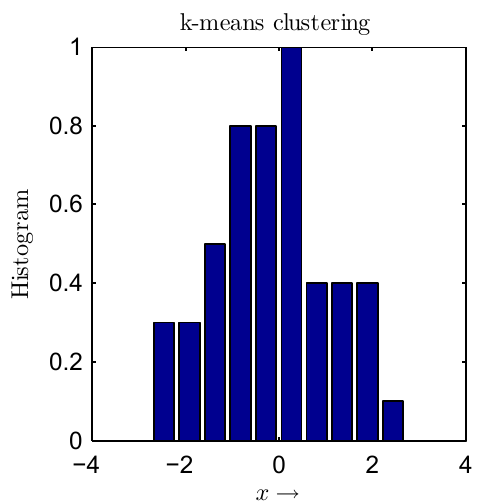}
	\caption{Normalized histograms of projections onto $x$-axis for reducing $M=5000$ samples of a two-dimensional standard normal distribution to $L=50$ samples. (left) Marginal of original samples. (middle) Marginal of LCD reduction result. (right) Marginal of result of k-means clustering.}
	\label{Fig_Histograms}
\end{figure*}

%
%

We now consider $M=100$ deterministic samples of a \SND{} shown in \Fig{Fig_Outlier}. The samples are calculated with the method from \cite{CDC09_HanebeckHuber}. One sample is replaced with an outlier located at $[3.5, 3.5]^T$. The point set is reduced to $L=10$ samples. The left side shows the result of the LCD reduction. The samples are well placed and only slightly shifted due to the outlier. On the right side, k-means clustering produces a result heavily disturbed by the outlier. In fact, one sample of the reduced point set is placed directly on the outlier, which significantly changes the mass distribution. Instead of representing $1$ \% of the distribution as before the reduction, the outlier now allocates $10$ \%.

\begin{figure*}[t]
	\hspace*{\fill}
	\includegraphics{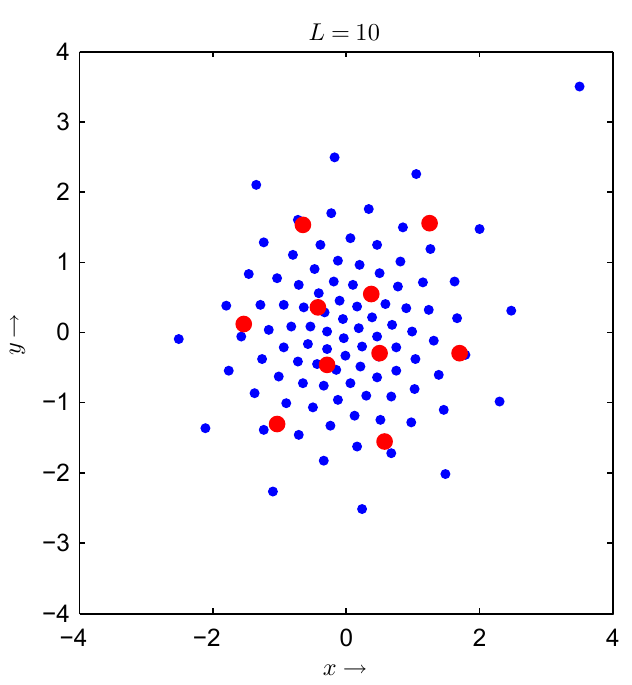}
	\hspace*{\fill}
	\includegraphics{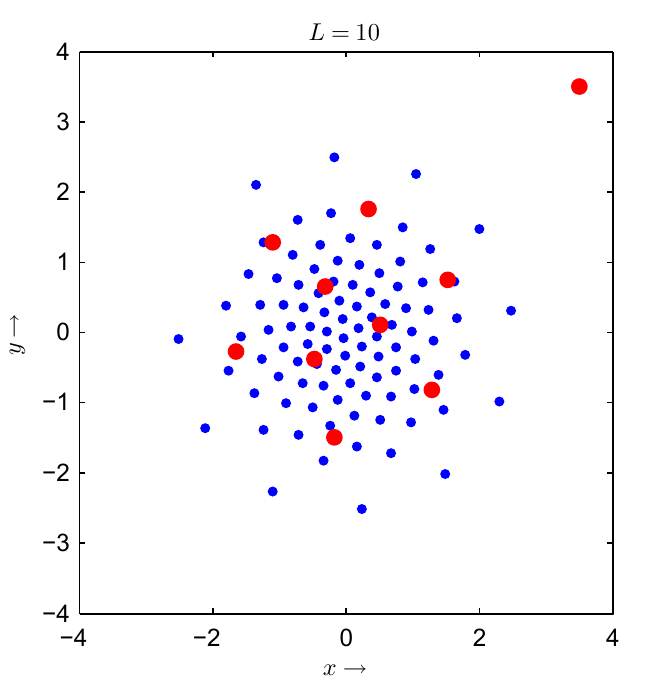}
	\hspace*{\fill}
	\caption{Blue: Deterministic samples representing a \SND{}. One sample is replaced by an outlier at $[3.5, 3.5]^T$. Red: Reduced point set. (left) LCD reduction. (right) k-means clustering.}
	\label{Fig_Outlier}
\end{figure*}

%
%

Finally, we investigate the robustness of the reduction methods with respect to missing data. For that purpose, we generate $2500$ samples and remove samples located within three vertical strips, see \Fig{Fig_Zebra}. The remaining samples are reduced down to $L=25$ samples. \Fig{Fig_Zebra} left shows the result of the LCD reduction, which almost gives the same results as before. The right side shows the result of k-means clustering, where it is obvious that samples are more or less placed along lines and the original mass distribution is not well maintained.

\begin{figure*}[t]
	\hspace*{\fill}
	\includegraphics{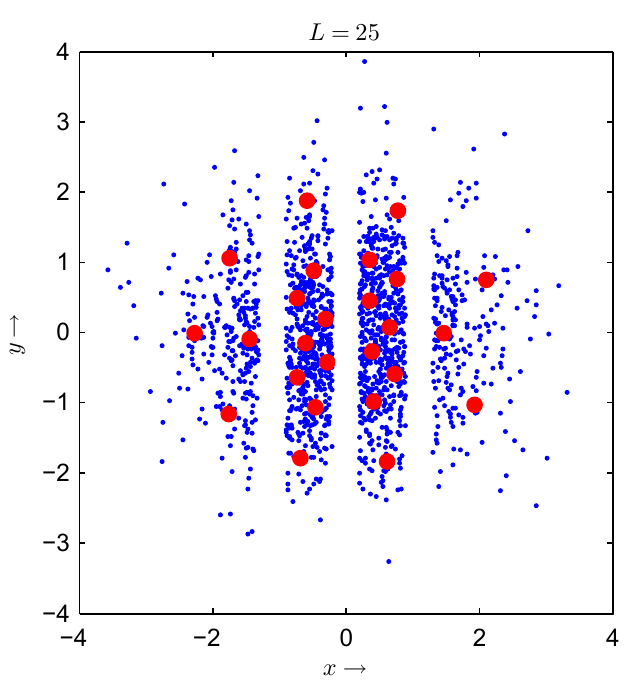}
	\hspace*{\fill}
	\includegraphics{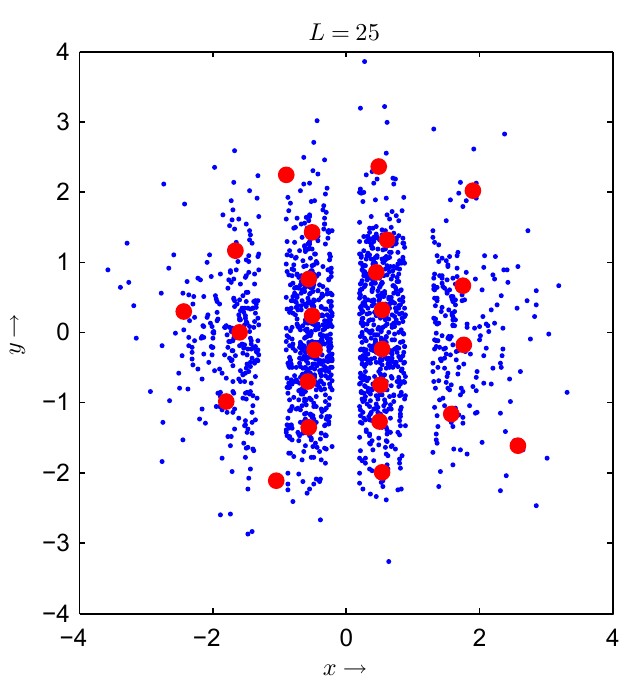}
	\hspace*{\fill}
	\caption{Blue: Random samples representing a \SND{} with some samples removed along three vertical lines. Red: Reduced point set. (left) LCD reduction. (right) k-means clustering.}
	\label{Fig_Zebra}
\end{figure*}


\section{Discussion} \label{Sec_Conclusions}

    %
%

A systematic approach for approximating a given Dirac mixture density by another one with less components has been introduced that is radically different from current clustering or vector quantization approaches. The (weights and) locations of the approximating density are calculated by minimizing a global distance measure, a generalization of the well-known \CvMD{} to the multivariate case. This generalization is obtained by defining an alternative to the classical cumulative distribution, the \LCD{} (LCD), as a characterization of discrete random quantities, which is unique and symmetric also in the multivariate case. 

%
%

Although kernels are used to define the LCD, this is not a kernel method. The distance measure is obtained by integrating over all possible kernels with all locations and widths, so that the final expression does not contain any kernel. 

%
%

The given Dirac mixture might be the result from random sampling or from certain processing steps involving analytic Dirac mixtures. In any case, the resulting approximating Dirac mixture is fully deterministic and the optimization process gives reproducible results.

%
%

Compared to clustering methods that find cluster heads minimizing the distance to their nearest neighbors, which is a local method, the LCD reduction globally matches the mass distributions of the given point set and its approximation. This leads to a smooth distance measure with almost no local minima that can be efficiently minimized with standard optimization procedures. However, it is important to note that due to its operating principle the proposed reduction method does not provide a mapping from old components to new components.

%
%

Constraints on the state variables can easily be considered when performing the approximation of the given density. An obvious application is the explicit avoidance of certain regions in the state space in order to obey certain physical constraints. However, an even more interesting application is the interpretation of the measurement equation in the Bayesian filtering step as an equality constraint for the state variables once an actual observation is available. This opens the way for more advanced filtering techniques in the case of Dirac mixture densities other than reweighing the individual component by the likelihood function.

%
%

Large data sets occur when performing Dirac mixture based state estimation in high--dimensional spaces or when considering product spaces of Dirac mixture densities. For a very large number of components, the computational effort for performing a direct reduction might be too large. For coping with this complexity issue, the proposed approach offers the unique feature of hierarchical approximation. For that purpose, the data set is decomposed into several smaller sets that are individually approximated. The resulting \DC s of the individual approximations are then collected into a single approximating Dirac mixture, which subsequently is further approximated to yield the desired number of components. Of course, this approximation hierarchy may consist of more intermediate approximation steps.


\section*{Acknowledgment}

	The author would like to thank Dipl.-Inform.~Henning Eberhardt for many fruitful discussions on this topic and the nice ideas for visualizing the performance of the proposed new reduction algorithm.

\bibliographystyle{StyleFiles/IEEEtran_Capitalize}

\bibliography{%
Literature,%
ISAS-Bibtex/ISASPublikationen,%
ISAS-Bibtex/ISASPublikationen_laufend,%
ISAS-Bibtex/ISASPreprints%
}

\begin{thebibliography}{10}
\providecommand{\url}[1]{#1}
\csname url@samestyle\endcsname
\providecommand{\newblock}{\relax}
\providecommand{\bibinfo}[2]{#2}
\providecommand{\BIBentrySTDinterwordspacing}{\spaceskip=0pt\relax}
\providecommand{\BIBentryALTinterwordstretchfactor}{4}
\providecommand{\BIBentryALTinterwordspacing}{\spaceskip=\fontdimen2\font plus
\BIBentryALTinterwordstretchfactor\fontdimen3\font minus
  \fontdimen4\font\relax}
\providecommand{\BIBforeignlanguage}[2]{{%
\expandafter\ifx\csname l@#1\endcsname\relax
\typeout{** WARNING: IEEEtran.bst: No hyphenation pattern has been}%
\typeout{** loaded for the language `#1'. Using the pattern for}%
\typeout{** the default language instead.}%
\else
\language=\csname l@#1\endcsname
\fi
#2}}
\providecommand{\BIBdecl}{\relax}
\BIBdecl

\bibitem{lloyd_least_1982}
S.~Lloyd, ``Least squares quantization in {PCM},'' \emph{{IEEE} Transactions on
  Information Theory}, vol.~28, no.~2, pp. 129--137, Mar. 1982.

\bibitem{macqueen_methods_1967}
\BIBentryALTinterwordspacing
J.~MacQueen, ``\BIBforeignlanguage{{EN}}{Some methods for classification and
  analysis of multivariate observations}.''\hskip 1em plus 0.5em minus
  0.4em\relax The Regents of the University of California, 1967. [Online].
  Available: \url{http://projecteuclid.org/euclid.bsmsp/1200512992}
\BIBentrySTDinterwordspacing

\bibitem{linde_algorithm_1980}
Y.~Linde, A.~Buzo, and R.~Gray, ``An Algorithm for Vector Quantizer Design,''
  \emph{{IEEE} Transactions on Communications}, vol.~28, no.~1, pp. 84--95,
  Jan. 1980.

\bibitem{crouse_look_2011}
D.~Crouse, P.~Willett, K.~Pattipati, and L.~Svensson, ``A look at Gaussian
  mixture reduction algorithms,'' in \emph{2011 Proceedings of the 14th
  International Conference on Information Fusion ({FUSION})}, Jul. 2011, pp.
  1--8.

\bibitem{williams_cost-function-based_2003}
J.~Williams and P.~Maybeck, ``Cost-Function-Based Gaussian Mixture Reduction
  for Target Tracking,'' in \emph{Proceedings of the Sixth International
  Conference of Information Fusion, 2003}, vol.~2, Jul. 2003, pp. 1047--1054.

\bibitem{runnalls_kullback-leibler_2007}
A.~Runnalls, ``Kullback-Leibler Approach to Gaussian Mixture Reduction,''
  \emph{{IEEE} Transactions on Aerospace and Electronic Systems}, vol.~43,
  no.~3, pp. 989--999, Jul. 2007.

\bibitem{MFI08_Hanebeck-LCD}
U.~D. Hanebeck and V.~Klumpp, ``{Localized Cumulative Distributions and a
  Multivariate Generalization of the Cram\'{e}r-von Mises Distance},'' in
  \emph{Proceedings of the 2008 IEEE International Conference on Multisensor
  Fusion and Integration for Intelligent Systems (MFI 2008)}, Seoul, Republic
  of Korea, Aug. 2008, pp. 33--39.

\bibitem{west_approximating_1993}
\BIBentryALTinterwordspacing
M.~West, ``Approximating Posterior Distributions by Mixture,'' \emph{Journal of
  the Royal Statistical Society. Series B (Methodological)}, vol.~55, no.~2,
  pp. 409--422, Jan. 1993. [Online]. Available:
  \url{http://www.jstor.org/stable/2346202}
\BIBentrySTDinterwordspacing

\bibitem{lefebvre_linear_2005}
T.~Lefebvre, H.~Bruyninckx, and J.~De~Schutter, ``The Linear Regression Kalman
  Filter,'' in \emph{Nonlinear Kalman Filtering for Force-Controlled Robot
  Tasks}, ser. Springer Tracts in Advanced Robotics, 2005, vol.~19.

\bibitem{julier_new_2000}
S.~Julier, J.~Uhlmann, and H.~F. Durrant-Whyte,
  ``\BIBforeignlanguage{English}{A New Method for the Nonlinear Transformation
  of Means and Covariances in Filters and Estimators},''
  \emph{\BIBforeignlanguage{English}{{IEEE} Transactions on Automatic
  Control}}, vol.~45, no.~3, pp. 477--482, Mar. 2000.

\bibitem{julier_scaled_2002}
S.~J. Julier, ``\BIBforeignlanguage{English}{The Scaled Unscented
  Transformation},'' in \emph{\BIBforeignlanguage{English}{Proceedings of the
  2002 {IEEE} American Control Conference ({ACC} 2002)}}, vol.~6, Anchorage,
  Alaska, {USA}, May 2002, pp. 4555-- 4559.

\bibitem{tenne_higher_2003}
D.~Tenne and T.~Singh, ``\BIBforeignlanguage{English}{The Higher Order
  Unscented Filter},'' in \emph{\BIBforeignlanguage{English}{Proceedings of the
  2003 {IEEE} American Control Conference ({ACC} 2003)}}, vol.~3, Denver,
  Colorado, {USA}, Jun. 2003, pp. 2441--2446.

\bibitem{IFAC08_Huber}
M.~F. Huber and U.~D. Hanebeck, ``{Gaussian Filter based on Deterministic
  Sampling for High Quality Nonlinear Estimation},'' in \emph{Proceedings of
  the 17th IFAC World Congress (IFAC 2008)}, vol.~17, no.~2, Seoul, Republic of
  Korea, Jul. 2008.

\bibitem{ACC13_Kurz}
G.~Kurz, I.~Gilitschenski, and U.~D. Hanebeck, ``{Recursive Nonlinear Filtering
  for Angular Data Based on Circular Distributions},'' in \emph{Proceedings of
  the 2013 American Control Conference (ACC 2013)}, Washington D. C., USA, Jun.
  2013.

\bibitem{Fusion13_Gilitschenski}
I.~Gilitschenski, G.~Kurz, and U.~D. Hanebeck, ``{Bearings-Only Sensor
  Scheduling Using Circular Statistics},'' in \emph{Proceedings of the 16th
  International Conference on Information Fusion (Fusion 2013)}, Istanbul,
  Turkey, Jul. 2013.

\bibitem{IPIN13_Kurz}
G.~Kurz, F.~Faion, and U.~D. Hanebeck, ``{Constrained Object Tracking on
  Compact One-dimensional Manifolds Based on Directional Statistics},'' in
  \emph{Proceedings of the Fourth IEEE GRSS International Conference on Indoor
  Positioning and Indoor Navigation (IPIN 2013)}, Montbeliard, France, Oct.
  2013.

\bibitem{arXiv14_Hanebeck}
\BIBentryALTinterwordspacing
U.~D. Hanebeck, ``{Truncated Moment Problem for Dirac Mixture Densities with
  Entropy Regularization},'' \emph{arXiv preprint: Systems and Control
  (cs.SY)}, Aug. 2014. [Online]. Available:
  \url{http://arxiv.org/abs/1408.7083}
\BIBentrySTDinterwordspacing

\bibitem{CDC06_Schrempf-DiracMixt}
O.~C. Schrempf, D.~Brunn, and U.~D. Hanebeck, ``{D}ensity {A}pproximation
  {B}ased on {D}irac {M}ixtures with {R}egard to {N}onlinear {E}stimation and
  {F}iltering,'' in \emph{Proceedings of the 2006 IEEE Conference on Decision
  and Control (CDC 2006)}, San Diego, California, USA, Dec. 2006.

\bibitem{MFI06_Schrempf-CramerMises}
------, ``{Dirac Mixture Density Approximation Based on Minimization of the
  Weighted Cram\'{e}r-von Mises Distance},'' in \emph{Proceedings of the 2006
  IEEE International Conference on Multisensor Fusion and Integration for
  Intelligent Systems (MFI 2006)}, Heidelberg, Germany, Sep. 2006, pp.
  512--517.

\bibitem{CDC09_HanebeckHuber}
U.~D. Hanebeck, M.~F. Huber, and V.~Klumpp, ``{Dirac Mixture Approximation of
  Multivariate Gaussian Densities},'' in \emph{Proceedings of the 2009 IEEE
  Conference on Decision and Control (CDC 2009)}, Shanghai, China, Dec. 2009.

\bibitem{ACC13_Gilitschenski}
I.~Gilitschenski and U.~D. Hanebeck, ``{Efficient Deterministic Dirac Mixture
  Approximation},'' in \emph{Proceedings of the 2013 American Control
  Conference (ACC 2013)}, Washington D. C., USA, Jun. 2013.

\bibitem{ACC10_Eberhardt}
H.~Eberhardt, V.~Klumpp, and U.~D. Hanebeck, ``{Optimal Dirac Approximation by
  Exploiting Independencies},'' in \emph{{Proceedings of the 2010 American
  Control Conference (ACC 2010)}}, Baltimore, Maryland, USA, Jun. 2010.

\bibitem{press_numerical_1992}
W.~H. Press, S.~A. Teukolsky, W.~T. Vetterling, and B.~P. Flannery,
  \emph{Numerical Recipes in C (2Nd Ed.): The Art of Scientific
  Computing}.\hskip 1em plus 0.5em minus 0.4em\relax New York, {NY}, {USA}:
  Cambridge University Press, 1992.

\end{thebibliography}


\begin{appendix}
	
	\section{Gradient of Distance Measure} \label{Sec_Gradient}
	
		%
%

Taking the derivative of the distance measure in \Eq{Eq_Distance_Measure_bmax} with respect to a location $x_\xi^{(\eta)}$ gives
\begin{equation*}
G_\xi^{(\eta)} = \frac{\partial D}{\partial x_\xi^{(\eta)}}
= - 2 \int _0^{b_\text{max}} \frac{1}{b^{N-1}} \int _{\NewR^N} 
\left( \tilde{F}(\vec{m}, b) - F(\vec{m}, b) \right ) \frac{\partial F(\vec{m}, b)}{\partial x_\xi^{(\eta)}} \, d \vec{m} \, d b \enspace ,
\end{equation*}
with
\begin{equation*}
\frac{\partial F(\vec{m}, b)}{\partial x_\xi^{(\eta)}} = 
- w_\xi^x \, \frac{x_\xi^{(\eta)} - m^{(\eta)}}{b^2} \prod_{k=1}^N \exp\left( -\frac{1}{2} \frac{ \left(x_\xi^{(k)}-m^{(k)} \right)^2 }{b^2} \right) 
\enspace .
\end{equation*}

In the following, the two parts of $G_\xi^{(\eta)}$ will be treated seperatly according to
\begin{equation}
G_\xi^{(\eta)} = G_\xi^{(\eta,1)} - G_\xi^{(\eta,2)} \enspace{.}
\label{Eq_G_G1_G2}
\end{equation}
The first part is given by
\begin{equation*}
G_\xi^{(\eta,1)} = - 2 \int _0^{b_\text{max}} \frac{1}{b^{N-1}} \int _{\NewR^N} 
\tilde{F}(\vec{m}, b) \frac{\partial F(\vec{m}, b)}{\partial x_\xi^{(\eta)}} \, d \vec{m} \, d b \enspace .
\end{equation*}
By using the expressions for $\tilde{F}(\vec{m}, b)$ and $\frac{\partial F(\vec{m}, b)}{\partial x_\xi^{(\eta)}}$, we obtain
\begin{equation*}
\begin{split}
G_\xi^{(\eta,1)} = 2 w_\xi^x \, \int _0^{b_\text{max}} \frac{1}{b^{N-1}} \int _{\NewR^N}
\frac{x_\xi^{(\eta)} - m^{(\eta)}}{b^2} & \prod_{k=1}^N \exp\left( -\frac{1}{2} \frac{ \left(x_\xi^{(k)}-m^{(k)} \right)^2 }{b^2} \right) \\
\sum_{i=1}^M w_i^{y} & \prod_{k=1}^N \exp\left( -\frac{1}{2} \frac{(y_i^{(k)}-m^{(k)})^2}{b^2} \right) d \vec{m} \, d b
\end{split}
\enspace .
\end{equation*}
Combining the product terms gives
\begin{equation*}
\begin{split}
G_\xi^{(\eta,1)} = 2 \, w_\xi^x \sum_{i=1}^M w_i^{y} \int _0^{b_\text{max}} \frac{1}{b^{N-1}}
\int _{\NewR} \frac{x_\xi^{(\eta)} - m^{(\eta)}}{b^2} & \exp\left( -\frac{1}{2} \frac{ \left(x_\xi^{(\eta)}-m^{(\eta)} \right)^2 }{b^2} \right) \\
& \exp\left( -\frac{1}{2} \frac{(y_i^{(\eta)}-m^{(\eta)})^2}{b^2} \right) d m^{(\eta)} \\
\prod_{\stackrel{k=1}{k \neq \eta}}^N \int _{\NewR} & \exp\left( -\frac{1}{2} \frac{ \left(x_\xi^{(k)}-m^{(k)} \right)^2 }{b^2} \right) \\
& \exp\left( -\frac{1}{2} \frac{(y_i^{(k)}-m^{(k)})^2}{b^2} \right) d m^{(k)} d b
\end{split}
\enspace .
\end{equation*}

For further simplification, the equality
\begin{equation*}
\begin{gathered}
\int _{\NewR} \frac{z_i-m}{b^2} 
\exp\left( -\frac{1}{2} \frac{(z_i - m)^2}{b^2} \right) \exp\left( -\frac{1}{2} \frac{(z_j - m)^2}{b^2} \right) d m \\
= \sqrt{\pi} \, \frac{z_i - z_j}{2 \, b} \exp \left( -\frac{1}{2} \frac{(z_i - z_j)^2}{2 \, b^2} \right)
\end{gathered}
\end{equation*}
is used together with the equality \Eq{Eq_Simplification_Product_of_Exp}, which leads to
\begin{equation*}
\begin{split}
G_\xi^{(\eta,1)} = & \pi^{\frac{N}{2}} \, w_\xi^x \sum_{i=1}^M w_i^{y} \left( x_\xi^{(\eta)} - y_i^{(\eta)} \right) \\
& \int _0^{b_\text{max}} \frac{1}{b} \prod_{k=1}^N \exp\left( -\frac{1}{2} \frac{ \left(x_\xi^{(k)}-y_i^{(k)} \right)^2 }{2 \, b^2} \right) d b
\end{split}
\end{equation*}
or equivalently to
\begin{equation*}
\begin{split}
G_\xi^{(\eta,1)} = & \pi^{\frac{N}{2}} \, w_\xi^x \sum_{i=1}^M w_i^{y} \left( x_\xi^{(\eta)} - y_i^{(\eta)} \right) \\
& \int _0^{b_\text{max}} \frac{1}{b} \exp\left( -\frac{1}{2} \frac{ \displaystyle\sum_{k=1}^N \left(x_\xi^{(k)}-y_i^{(k)} \right)^2 }{2 \, b^2} \right) d b
\end{split}
\enspace .
\end{equation*}

With
\begin{equation*}
\int _0^{b_\text{max}} \frac{1}{b} \, \exp\left( -\frac{1}{2} \frac{ z }{2 \, b^2} \right) \, d b
= - \frac{1}{2} \, \Ei\left( -\frac{1}{2} \frac{ z }{2 \, b_\text{max}^2}
 \right)
\end{equation*}
for $z>0$, the expression
\begin{equation*}
\frac{\partial D}{\partial x_\xi^{(\eta)}} = \frac{\pi^{\frac{N}{2}}}{2} w_\xi^{x}
\sum_{i=1}^L w_i^{x} \left( x_\xi^{(\eta)} - x_i^{(\eta)} \right) 
\Ei \left( -\frac{1}{2} \frac{ \displaystyle\sum_{k=1}^N \left( x_\xi^{(k)} - y_i^{(k)} \right)^2 }{2 b_\text{max}^2} \right)
\end{equation*}
is obtained for component index $\xi=1, \ldots, L$ and dimension index $k=1, \ldots, N$.

The second part is given by
\begin{equation*}
G_\xi^{(\eta,2)} = - 2 \int _0^{b_\text{max}} \frac{1}{b^{N-1}} \int _{\NewR^N} 
F(\vec{m}, b) \frac{\partial F(\vec{m}, b)}{\partial x_\xi^{(\eta)}} \, d \vec{m} \, d b \enspace .
\end{equation*}
By using the expressions for $F(\vec{m}, b)$ and $\frac{\partial F(\vec{m}, b)}{\partial x_\xi^{(\eta)}}$, we obtain
\begin{equation*}
\begin{split}
G_\xi^{(\eta,2)} = 2 w_\xi^x \, \int _0^{b_\text{max}} & \frac{1}{b^{N-1}} \int _{\NewR^N} \\
\frac{x_\xi^{(\eta)} - m^{(\eta)}}{b^2} & \prod_{k=1}^N \exp\left( -\frac{1}{2} \frac{ \left(x_\xi^{(k)}-m^{(k)} \right)^2 }{b^2} \right) \\
\sum_{i=1}^L w_i^{x} & \prod_{k=1}^N \exp\left( -\frac{1}{2} \frac{(x_i^{(k)}-m^{(k)})^2}{b^2} \right) d \vec{m} \, d b
\end{split}
\enspace .
\end{equation*}
Combining the product terms gives
\begin{equation*}
\begin{split}
G_\xi^{(\eta,2)} = 2 \, w_\xi^x 
\sum_{i=1}^L w_i^{x} \int _0^{b_\text{max}} \frac{1}{b^{N-1}}
\int _{\NewR} \frac{x_\xi^{(\eta)} - m^{(\eta)}}{b^2} & \exp\left( -\frac{1}{2} \frac{ \left(x_\xi^{(\eta)}-m^{(\eta)} \right)^2 }{b^2} \right) \\
& \exp\left( -\frac{1}{2} \frac{(x_i^{(\eta)}-m^{(\eta)})^2}{b^2} \right) d m^{(\eta)} \\
\prod_{\stackrel{k=1}{k \neq \eta}}^N \int _{\NewR} & \exp\left( -\frac{1}{2} \frac{ \left(x_\xi^{(k)}-m^{(k)} \right)^2 }{b^2} \right) \\
& \exp\left( -\frac{1}{2} \frac{(x_i^{(k)}-m^{(k)})^2}{b^2} \right) d m^{(k)} d b
\end{split}
\enspace .
\end{equation*}
For further simplification, we use
\begin{equation*}
\int _{\NewR} \frac{z_i-m}{b^2} 
\exp\left( -\frac{1}{2} \frac{(z_i - m)^2}{b^2} \right) \exp\left( -\frac{1}{2} \frac{(z_j - m)^2}{b^2} \right) d m
= \sqrt{\pi} \, \frac{z_i - z_j}{2 \, b} \exp \left( -\frac{1}{2} \frac{(z_i - z_j)^2}{2 \, b^2} \right)
\end{equation*}
and \Eq{Eq_Simplification_Product_of_Exp}, which leads to
\begin{equation*}
G_\xi^{(\eta,2)} = \pi^{\frac{N}{2}} \, w_\xi^x \sum_{i=1}^L w_i^{x} \left( x_\xi^{(\eta)} - x_i^{(\eta)} \right)
\int _0^{b_\text{max}} \frac{1}{b} \prod_{k=1}^N \exp\left( -\frac{1}{2} \frac{ \left(x_\xi^{(k)}-x_i^{(k)} \right)^2 }{2 \, b^2} \right) d b \enspace .
\end{equation*}
or equivalently
\begin{equation*}
G_\xi^{(\eta,2)} = \pi^{\frac{N}{2}} \, w_\xi^x \sum_{i=1}^L w_i^{x} \left( x_\xi^{(\eta)} - x_i^{(\eta)} \right)
\int _0^{b_\text{max}} \frac{1}{b} \exp\left( -\frac{1}{2} \frac{ \displaystyle\sum_{k=1}^N \left(x_\xi^{(k)}-x_i^{(k)} \right)^2 }{2 \, b^2} \right) d b 
\enspace .
\end{equation*}
With
\begin{equation*}
\int _0^{b_\text{max}} \frac{1}{b} \, \exp\left( -\frac{1}{2} \frac{ z }{2 \, b^2} \right) \, d b
= - \frac{1}{2} \, \Ei\left( -\frac{1}{2} \frac{ z }{2 \, b_\text{max}^2}
 \right)
\end{equation*}
for $z>0$, we obtain
\begin{equation*}
G_\xi^{(\eta,2)} = \frac{\pi^{\frac{N}{2}}}{2} w_\xi^{x}
- \sum_{i=1}^M w_i^{y} \left( x_\xi^{(\eta)} - y_i^{(\eta)} \right) 
\Ei \left( -\frac{1}{2} \frac{ \displaystyle\sum_{k=1}^N \left( x_\xi^{(k)} - x_i^{(k)} \right)^2 }{2 b_\text{max}^2} \right)
\end{equation*}
for component index $\xi=1, \ldots, L$ and dimension index $k=1, \ldots, N$.

By combining the two results for $G_\xi^{(\eta,1)}$ and $G_\xi^{(\eta,2)}$ according to \Eq{Eq_G_G1_G2}, we obtain the following Theorem.

\begin{theorem}
The gradient of the general distance measure in Theorem~\ref{Theorem_General_Distance_Measure} with respect to the locations of the Dirac components is given by
\begin{equation*}
\begin{split}
\frac{\partial D}{\partial x_\xi^{(\eta)}} = \frac{\pi^{\frac{N}{2}}}{2} w_\xi^{x}
\Bigg\{ \sum_{i=1}^L & w_i^{x} \left( x_\xi^{(\eta)} - y_i^{(\eta)} \right) 
\Ei \left( -\frac{1}{2} \frac{ \displaystyle\sum_{k=1}^N \left( x_\xi^{(k)} - y_i^{(k)} \right)^2 }{2 b_\text{max}^2} \right) \\
- \sum_{i=1}^M & w_i^{y} \left( x_\xi^{(\eta)} - x_i^{(\eta)} \right) 
\Ei \left( -\frac{1}{2} \frac{ \displaystyle\sum_{k=1}^N \left( x_\xi^{(k)} - x_i^{(k)} \right)^2 }{2 b_\text{max}^2} \right) \Bigg\}
\end{split}
\end{equation*}
for component index $j=1, \ldots, L$ and dimension index $k=1, \ldots, N$.
\label{Theorem_Exact_Gradient}
\end{theorem}

For large $b_\text{max}$, the $\Ei$--function in Theorem~\ref{Theorem_Exact_Gradient} can be approximated according to \Eq{Eq_Ei_Approx}. Hence, we have
\begin{equation*}
\begin{split}
\Ei\left( -\frac{z}{4 \, b_\text{max}^2} \right) 
& \approx \Gamma - \frac{z}{4 \, b_\text{max}^2} + \log\left( \frac{z}{4 \, b_\text{max}^2} \right) \\
& \approx \Gamma - \log\left( 4 \, b_\text{max}^2 \right) + \log(z) \\
& = - C_b + \log(z)
\end{split}
\end{equation*}
for $z>0$. With
\begin{equation*}
\begin{gathered}
- C_b \left\{ \sum_{i=1}^L w_i^{x} \left( x_\xi^{(\eta)} - x_i^{(\eta)} \right) 
- \sum_{i=1}^M w_i^{y} \left( x_\xi^{(\eta)} - y_i^{(\eta)} \right) \right\} \\
= C_b \left( \sum_{i=1}^L w_i^{x} x_i^{(\eta)} - \sum_{i=1}^M w_i^{y} y_i^{(\eta)} \right) \enspace ,
\end{gathered}
\end{equation*}
we obtain the next Theorem.
\begin{theorem} \label{Theorem_Approximate_Gradient}
For large $b_\text{max}$, the gradient of the distance measure with respect to the locations of the Dirac components is given by
\begin{equation}
\begin{split}
\frac{\partial D}{\partial x_\xi^{(\eta)}} = \frac{\pi^\frac{N}{2}}{2} w_\xi^{x}
& \left\{ \sum_{i=1}^L w_i^{x} \left( x_\xi^{(\eta)} - x_i^{(\eta)} \right) \log \left( \sum_{k=1}^N \left( x_\xi^{(k)} - x_i^{(k)} \right)^2 \right) \right.\\
- & \sum_{i=1}^M w_i^{y} \left( x_\xi^{(\eta)} - y_i^{(\eta)} \right) \log \left( \sum_{k=1}^N \left( x_\xi^{(k)} - y_i^{(k)} \right)^2 \right) \\
+ &  C_b \cdot \left. \left( \sum_{i=1}^L w_i^{x} x_i^{(\eta)} - \sum_{i=1}^M w_i^{y} y_i^{(\eta)} \right) \right\} \enspace ,
\end{split}
\label{Eq_Approximate_Gradient}
\end{equation}
for component index $\xi=1, \ldots, L$ and dimension index $\eta=1, \ldots, N$.
\end{theorem}

\begin{corollary}
For equal expected values, the optimal locations $\vec{x}_i=\begin{bmatrix} x_i, y_i\end{bmatrix}^T$ of the components $i=1, \ldots, L$ of the approximating Dirac mixture density are obtained by solving the following $N \cdot L$ equations (necessary conditions)
\begin{equation}
\begin{split}
  & \sum_{i=1}^M w_i^{y} \left( x_\xi^{(\eta)} - y_i^{(\eta)} \right) \log \left( \sum_{k=1}^N \left( x_\xi^{(k)} - y_i^{(k)} \right)^2 \right) \\
= & \sum_{i=1}^L w_i^{x} \left( x_\xi^{(\eta)} - x_i^{(\eta)} \right) \log \left( \sum_{k=1}^N \left( x_\xi^{(k)} - x_i^{(k)} \right)^2 \right)
\enspace ,
\end{split}
\label{Corollary_System_of_Equations}
\end{equation}
for component index $\xi=1, \ldots, L$ and dimension index $\eta=1, \ldots, N$.
\end{corollary}
	
\end{appendix}

\end{document}